\documentclass[11pt]{article}
\usepackage{amssymb}
\usepackage{amsmath}
\usepackage{enumerate,xspace,multirow,theorem}
\usepackage{graphicx}	

\newtheorem{theorem}{Theorem}[section]
\newtheorem{corollary}[theorem]{Corollary}
\newtheorem{lemma}[theorem]{Lemma}

\newtheorem{definition}[theorem]{Definition}
{\theorembodyfont{\rmfamily} \newtheorem{remark}[theorem]{Remark}
\newtheorem{example}{Example}}

\def\blksquare{\rule{2mm}{2mm}}
\def\qedsymbol{\blksquare}
\newcommand{\bg}[1]{\medskip\noindent{\it #1}}
\newcommand{\ed}{{\hfill\qedsymbol}\medskip}
\newenvironment{proof}{\bg{Proof : }}{\ed}

\newcommand{\np}{{\em NP}\xspace}
\newcommand{\nphard}{\np-hard\xspace}

\newcommand{\E}[2][{}]{\ensuremath{\mathrm{E}_{#1}\bigl[#2\bigr]}}

\newcommand{\R}{\ensuremath{\mathbb R}}
\newcommand{\Z}{\ensuremath{\mathbb Z}}

\newcommand{\A}{\ensuremath{\mathcal{A}}}
\newcommand{\B}{\ensuremath{\mathcal{B}}}

\newcommand{\G}{\ensuremath{\mathcal{G}}}
\newcommand{\I}{\ensuremath{\mathcal I}}
\newcommand{\F}{\ensuremath{\mathcal F}}
\newcommand{\D}{\ensuremath{\mathcal D}}
\newcommand{\T}{\ensuremath{\mathcal T}}

\newcommand{\Oc}{\ensuremath{\mathcal O}}

\newcommand{\OPT}{\ensuremath{\mathit{OPT}}}

\newcommand{\argmin}{\ensuremath{\mathrm{argmin}}}

\newcommand{\frall}{\ensuremath{\text{ for all }}}

\newcommand{\sm}{\ensuremath{\setminus}}
\newcommand{\es}{\ensuremath{\emptyset}}
\newcommand{\assign}{\ensuremath{\leftarrow}}

\newcommand{\e}{\ensuremath{\epsilon}}

\newcommand{\gm}{\ensuremath{\gamma}}

\newcommand{\sse}{\subseteq}

\newcommand{\ufl}{{\textsf{UFL}}\xspace}
\newcommand{\sufl}{{\small \textsf{UFL}}\xspace}

\newcommand{\vcp}{{\textsf{VC}}\xspace}
\newcommand{\svcp}{{\small \textsf{VC}}\xspace}
\newcommand{\mvcp}{{\textsf{Multi-VC}}\xspace}

\newcommand{\hy}{\ensuremath{\hat y}}

\newcommand{\hc}{\ensuremath{\hat c}}
\newcommand{\hht}{\ensuremath{\hat t}}

\newcommand{\ld}{\ensuremath{\lambda}}

\newcommand{\al}{\ensuremath{\alpha}}

\newcommand{\dt}{\ensuremath{\delta}}
\newcommand{\Dt}{\ensuremath{\Delta}}
\newcommand{\sg}{\ensuremath{\sigma}}

\newcommand{\pub}{\ensuremath{\mathit{pub}}}
\newcommand{\bff}{\ensuremath{\overline f}}
\newcommand{\bc}{\ensuremath{\overline c}}
\newcommand{\into}{\ensuremath{\mathrm{in}}}

\title{Truthful Mechanism Design for Multidimensional Covering Problems%
\footnote{A preliminary version appeared as~\cite{MinooeiS12}. Theorem 13
in~\cite{MinooeiS12} is incorrect; the correct statements appear as Theorem~\ref{mclosed}
and Corollary~\ref{mclosed-3hop} here.}}

\author{
         Hadi Minooei\thanks{{\tt\{hminooei,cswamy\}@math.uwaterloo.ca}.
         Dept. of Combinatorics and Optimization, Univ. Waterloo, Waterloo, ON N2L 3G1.
         Research supported partially by NSERC grant 327620-09 and the second author's
         Ontario Early Researcher Award.}
\and
         Chaitanya Swamy$^{\text{\thefootnote}}$
} 

\date{} 

\begin{document}

\maketitle

\vspace*{-3ex}

\begin{abstract}
We investigate {\em multidimensional covering mechanism-design} problems, wherein there
are $m$ items that need to be covered and $n$ agents who provide covering objects, with
each agent $i$ having a private cost for the covering objects he provides. The goal is to
select a set of covering objects of minimum total cost that together cover all the items.

We focus on two representative covering problems: uncapacitated facility location (\ufl)
and vertex cover (\vcp). For multidimensional \ufl, we give a black-box method to
transform any {\em Lagrangian-multiplier-preserving} $\rho$-approximation algorithm for
\ufl to a truthful-in-expectation, $\rho$-approx. mechanism. This yields the first
result for multidimensional \ufl, namely a truthful-in-expectation 2-approximation
mechanism. 

For multidimensional \vcp (\mvcp), we develop a {\em decomposition method} that reduces
the mechanism-design problem into the simpler task of constructing 
{\em threshold mechanisms}, which are a restricted class of truthful mechanisms, for 
simpler (in terms of graph structure or problem dimension) instances of \mvcp. By 
suitably designing the decomposition and the threshold mechanisms it uses as building
blocks, we obtain truthful mechanisms with the following approximation ratios ($n$ is the
number of nodes):  
(1) $O(r^2\log n)$ for $r$-dimensional \vcp; and
(2) $O(r\log n)$ for $r$-dimensional \vcp on any proper minor-closed family of graphs
(which improves to $O(\log n)$ if no two neighbors of a node belong to the same player).
These are the first truthful mechanisms for \mvcp with non-trivial approximation
guarantees.    
\end{abstract}

\section{Introduction} \label{intro}
Algorithmic mechanism design (AMD) deals with efficiently-computable algorithmic
constructions in the presence of strategic players who hold the inputs to the problem, and 
may misreport their input if doing so benefits them.
The challenge is to design algorithms that work well with the true (privately-known) input.   
In order to achieve this task, a {\em mechanism} specifies both an algorithm and a
pricing or payment scheme that can be used to incentivize players to reveal their true
inputs. A mechanism is said to be {\em truthful}, if each player maximizes his utility by
revealing his true input regardless of the other players' declarations. 

In this paper, we initiate a study of {\em multidimensional covering mechanism-design}
problems, often called {\em reverse auctions} or {\em procurement auctions} in the
mechanism-design literature. 
These can be abstractly stated as follows. There are $m$ items that need to be 
covered and $n$ agents who provide covering objects, with each agent $i$ having a private
cost for the covering objects he provides. The goal is to select (or buy) a suitable set
of covering objects from each player so that their union covers all the items, and the
total covering cost incurred is minimized. This {\em cost-minimization} (CM) problem is 
equivalent to the {\em social-welfare maximization} (SWM) (where the social welfare is 
$-$ (total cost incurred by the players and the mechanism designer)), so
ignoring computational efficiency, the classical VCG
mechanism~\cite{Vickrey61,Clarke71,Groves73} yields a truthful mechanism that always
returns an optimal solution. However, the CM problem is often \nphard, so we seek to
design a {\em polytime} truthful mechanism where the underlying algorithm returns a
near-optimal solution to the CM problem.  

Although multidimensional packing mechanism-design problems have received much attention
in the AMD literature, multidimensional covering CM problems are conspicuous by their
absence in the literature. 
For example, the packing SWM problem of combinatorial auctions has been studied (in
various flavors) in numerous works both from the viewpoint of designing polytime truthful,
approximation mechanisms~\cite{DobzinskiNS12,LaviS11,Dobzinski08,DughmiRY11},
and from the perspective of proving lower bounds on the capabilities of computationally-
(or query-) efficient truthful
mechanisms~\cite{LaviMN03,DughmiV11,DobzinskiV12}.  
In contrast, the lack of study of multidimensional covering CM problems is aptly
summarized by the blank table entry for results on truthful approximations for
procurement auctions in Fig. 11.2 in~\cite{agt} (a recent result of~\cite{BBShaddin} is an
exception; see ``Related work'').  
In fact, to our knowledge, the only multidimensional problem with a covering flavor that
has been studied in the AMD literature is the makespan-minimization problem on unrelated
machines~\cite{NisanR01,LaviS09,AshlagiDL09}, which is not an SWM problem. 
 
\vspace{-0.5ex}
\paragraph{Our results and techniques.}
We study two representative multidimensional covering problems, namely 
(metric) {\em uncapacitated facility location (\ufl)}, 
and {\em vertex cover} (\vcp), 
and develop  various techniques to devise polytime, truthful, approximation mechanisms for 
these problems. 

For multidimensional \ufl (Section~\ref{fl}), wherein players own (known) different
facility sets and the assignment costs are public, we present a {\em black-box reduction
from truthful mechanism design to algorithm design}.  
We show that any $\rho$-approximation algorithm for \ufl satisfying an additional 
{\em Lagrangian-multiplier-preserving} (LMP) property (that indeed holds for various
algorithms) can be converted in a black-box fashion to a truthful-in-expectation 
$\rho$-approximation mechanism (Theorem~\ref{flthm}). This is the {\em first} such
black-box reduction for a multidimensional covering problem, and it leads to 
the first result for multidimensional \ufl, namely, a
truthful-in-expectation, 2-approximation mechanism. 
Our result builds upon the convex-decomposition technique in~\cite{LaviS11}. Lavi and
Swamy~\cite{LaviS11} primarily focus on packing problems, but remark that their
convex-decomposition idea also yields results for {\em single-dimensional} covering
problems, and leave open the problem of obtaining results for multidimensional covering
problems. Our result for \ufl identifies an interesting property under which a
$\rho$-approximation algorithm for a covering problem can be transformed into a truthful,
$\rho$-approximation mechanism in the multidimensional setting. 

In Section~\ref{vc}, we consider multidimensional \vcp, where each player owns a
(known) set of nodes.  
Although, algorithmically, \vcp is one of the simplest covering problems, 
it becomes a surprisingly challenging mechanism-design problem in the 
{\em multidimensional} mechanism-design setting, and, in fact, seems significantly more  
difficult than multidimensional \ufl.
This is in stark contrast with the single-dimensional setting, where each player owns a
single node.    
Before detailing our results and techniques, we mention some of the difficulties
encountered. 
We use \mvcp to distinguish the multidimensional mechanism-design problem from the
algorithmic problem. 

For {\em single-dimensional} problems, a simple monotonicity condition characterizes the 
{\em implementability} of an algorithm, that is, whether it can be combined with 
suitable payments to obtain a truthful mechanism. This condition allows for ample
flexibility and various algorithm-design techniques can be leveraged to design
monotone algorithms for both covering and packing problems (see,
e.g.,~\cite{BriestKV05,LaviS11}). 
For single-dimensional \vcp, many of the known 2-approximation
algorithms for the algorithmic problem (based on LP-rounding, primal-dual methods, or
combinatorial methods) 
are either already monotone, or can be modified in simple ways so that they become
monotone, and thereby yield truthful 2-approximation mechanisms~\cite{DevanurMV05}.  
However, the underlying algorithm-design techniques fail to yield
algorithms satisfying {\em weak monotonicity} (WMON)---a necessary condition for
implementability (see Theorem~\ref{prelim})---even for the simplest multidimensional
setting, namely, 2-dimensional \vcp, where {\em every player owns at most two nodes}.
We show this for various LP-rounding methods in Appendix~\ref{lpround-bad}, and 
for primal-dual algorithms in Appendix~\ref{other-bad}. 

Furthermore, various techniques that have been devised for designing polytime truthful
mechanisms for multidimensional packing problems (such as combinatorial auctions) do not
seem to be helpful for \mvcp. For instance, the well-known technique of constructing a
{\em maximal-in-range}, or more generally, a {\em maximal-in-distributional-range} (MIDR) 
mechanism---fix some subset of outcomes and return the best outcome in this
set---does not work for \mvcp~\cite{BBShaddin} (and more generally, for multidimensional
covering problems). (More precisely, any algorithm for \mvcp whose range is a proper
subset of the collection of minimal vertex covers, cannot have bounded approximation
ratio.)  This also rules out the convex-decomposition technique of~\cite{LaviS11}, which
we exploit for multidimensional \ufl, because, as noted in~\cite{LaviS11}, this yields an
MIDR mechanism.  

Thus, we need to develop new techniques to attack \mvcp (and multidimensional covering
problems in general). We devise two main techniques for \mvcp. 
We introduce a simple class of truthful mechanisms called {\em threshold mechanisms}
(Section~\ref{thresh}), and show that despite their restrictions, 
threshold mechanisms can achieve non-trivial 
approximation guarantees. We next develop a {\em decomposition method} for \mvcp
(Section~\ref{decomp}) that provides a general way of reducing the mechanism-design
problem for \mvcp into simpler---either in terms of graph structure, or problem
dimension---mechanism-design problems by using threshold mechanisms as building
blocks. 
We believe that these techniques will also find use in other mechanism-design problems.   

By leveraging the decomposition method along with threshold mechanisms, we obtain various
truthful, approximation mechanisms for \mvcp, which yield the {\em first} truthful
mechanisms for multidimensional vertex cover with non-trivial approximation guarantees. 
Let $n$ be the number of nodes.
Our decomposition method shows that any instance of $r$-dimensional \vcp can be broken
up into $O(r^2\log n)$ instances of {\em single-dimensional \vcp}; this in turn leads to a
truthful, $O(r^2\log n)$-approximation mechanism for $r$-dimensional \vcp
(Theorem~\ref{rdim}). In particular, for any fixed $r$, we obtain an 
$O(\log n)$-approximation for any graph. 
We give another decomposition method that yields an improved truthful, 
$O(r\log n)$-approximation mechanism (Theorem~\ref{mclosed}) for any proper minor-closed
family of graphs (such as planar graphs). This guarantee improves to $O(\log n)$ for any
proper minor-closed family, when no two neighbors of a node belong to the same player.

It is worthwhile to note that in addition to their usefulness in the design
of truthful, approximation mechanisms for \mvcp, some of the mechanisms we design
also enjoy good frugality properties. 
We obtain (Theorem~\ref{frugality}) the {\em first} mechanisms for \mvcp that are
polytime, truthful and {\em simultaneously} achieve bounded approximation ratio {\em and}
bounded frugality ratio with respect to the benchmarks
in~\cite{ElkindFOCS,KempeFOCS}. This nicely complements a result of~\cite{ElkindFOCS}, 
who devise such a mechanism for single-dimensional \vcp. 

\paragraph{Related work.}
As mentioned earlier, there is little prior work on the CM problem for multidimensional
covering problems. 
Dughmi and Roughgarden~\cite{BBShaddin} give a general technique to convert an FPTAS for
an SWM problem to a truthful-in-expectation FPTAS. However, for covering problems, they
obtain an additive approximation, which does not translate to a (worst-case)
multiplicative approximation. In fact, as they observe, a multiplicative 
approximation ratio is impossible (in polytime) using their technique, or any other
technique that constructs a MIDR mechanism whose range is a proper subset of all
outcomes. 

For single-dimensional covering problems, various other results, including black-box
results, are known. 
Briest et al.~\cite{BriestKV05} consider a closely-related generalization, which one may
call the ``single-value setting''; although this is a multidimensional setting, it admits
a simple monotonicity condition sufficient for implementability, which makes this 
setting easier to deal with than our multidimensional settings.
They show that a pseudopolynomial time algorithm (for covering and packing problems)
can be converted into a truthful FPTAS. 
Lavi and Swamy~\cite{LaviS11} mainly consider packing problems, but mention that their
technique also yields results for single-dimensional covering problems.

Single-dimensional covering problems have been well studied from the perspective of 
{\em frugality}. Here the goal is to design mechanisms that have bounded (over-)payment
with respect to some benchmark, but one does not (typically) care about the cost of the
solution returned. Starting with the work of Archer and Tardos~\cite{ArcherT07},
various benchmarks for frugality have been proposed and investigated for various problems
including \vcp, $k$-edge-disjoint paths, spanning tree, $s$-$t$ cut; 
see~\cite{Karlin,Elkind1,KempeFOCS,ElkindFOCS} and the references therein. 
Some of our mechanisms for \mvcp are inspired by the constructions
in~\cite{KempeFOCS,ElkindFOCS}, and simultaneously achieve bounded approximation ratio and
bounded frugality ratio.

Our decomposition method, where we combine mechanisms for simpler problems into a
mechanism for the given problem, is somewhat in the same spirit as the construction
in~\cite{MualemN08}. They give a toolkit for combining truthful mechanisms, identifying
sufficient conditions under which this combination preserves truthfulness. But they work
only with the single-dimensional setting, which is much more tractable to deal with. 

Finally, as noted earlier, there are a wide variety of results on truthful
mechanism-design for packing SWM problems, such as combinatorial
auctions~\cite{DobzinskiNS12,LaviS11,Dobzinski08,DughmiRY11,LaviMN03,DughmiV11,DobzinskiV12}. 

\section{Preliminaries} \label{prelim}
In a {\it multidimensional covering mechanism-design problem}, we have $m$ items that
need to be covered, and $n$ agents/players who provide covering objects. Each agent $i$
provides a set $\T_i$ of covering objects. 
All this information is public knowledge. 
We use $[k]$ to denote the set $\{1,\ldots,k\}$. 
Each agent $i$ has a {\em private cost} (or type) vector $c_i=\{c_{i,v}\}_{v\in\T_i}$,
where $c_{i,v}$ is the cost he incurs for providing object $v\in\T_i$; 
for $T\sse\T_i$, we use $c_i(T)$ to denote $\sum_{v\in T}c_{i,v}$. 
A feasible solution or allocation selects a subset $T_i\sse\T_i$ for each agent $i$,  
denoting that $i$ provides the objects in $T_i$. 
Given this solution, each agent $i$ incurs the private cost $c_i(T_i)$. 
Also, the mechanism designer incurs a publicly-known cost $\pub(T_1,\ldots,T_n)$.
The goal is to minimize the total cost $\sum_i c_i(T_i)+\pub(T_1,\ldots,T_n)$ incurred. 
We call this the {\em cost minimization} (CM) problem. 
Note that we can encode any feasibility constraints in the covering problem by simply
setting $\pub(a)=\infty$ if $a$ is not a feasible allocation.  
Observe that if we view the mechanism designer also as a player, then 
the CM problem is equivalent to maximizing the social welfare, which is given
by $\sum_i-c_i(T_i)-\pub(T_1,\ldots,T_n)$. 

Various covering problems can be cast in the above framework. 
For example, in the mechanism-design version of {\em vertex cover} (Section~\ref{vc}), 
the items are edges of a graph. 
Each agent $i$ provides a subset $\T_i$ of the nodes of the graph 
and incurs a private cost $c_{i,v}$ if node $v\in T_i$ is used to cover an edge.  
We can set $\pub(T_1,\ldots,T_n)=0$ if $\bigcup_i T_i$ is a vertex cover, 
and $\infty$ otherwise, to encode that the solution must be a vertex cover.
It is also easy to see that the mechanism-design version of 
{\em uncapacitated facility location} (\ufl; Section~\ref{fl}), where each agent provides
some facilities and has private facility-opening costs, and the client-assignment costs
are public, can be modeled by letting $\pub(T_1,\ldots,T_n)$ be the total
client-assignment cost given the set $\bigcup_i T_i$ of open facilities.  

Let $C_i$ denote the set of all possible cost functions of agent $i$, and 
$\Oc$ be the (finite) set of all possible allocations. Let $C=\prod_{i=1}^n C_i$. 
For a tuple $x=(x_1,\ldots,x_n)$, we use $x_{-i}$ to denote
$(x_1,\ldots,x_{i-1},x_{i+1},\ldots,x_n)$. Similarly, let $C_{-i}=\prod_{j\neq i}C_j$.  
For an allocation $a=(T_1,\ldots,T_n)$, we sometimes use $a_i$ to
denote $T_i$, $c_i(a)$ to denote $c_i(a_i)=c_i(T_i)$. 
A (direct revelation) {\em mechanism} $M=({\cal A},p_1,\ldots,p_n)$ for a covering problem
consists of 
an allocation algorithm ${\cal A}:C\mapsto\Oc$ and a payment function $p_i:C\mapsto\R$ for
each agent $i$, and works as follows. Each agent $i$ reports a cost function $c_i$ (that
might be different from his true cost function). 
The mechanism 
computes the allocation $\A(c)=(T_1,\ldots,T_n)$, and pays $p_i(c)$ to each agent
$i$. Throughout, we use $\bc_i$ to denote the true cost 
function of $i$. The {\it utility} $u_i(c_i,c_{-i};\bc_i)$ that player $i$ derives when
he reports $c_i$ and the others report $c_{-i}$ is $p_i(c)-\bc_i(T_i)$, and each agent $i$
aims to maximize his own utility (rather than the social welfare). 

A desirable property for a mechanism to satisfy is {\em truthfulness}, wherein
every agent $i$ maximizes his utility by reporting his true cost function.
All our mechanisms will also satisfy the natural property of 
{\em individual rationality} (IR), which 
means that every agent has nonnegative utility if he reports his true cost.  

\begin{definition}
A mechanism $M=\bigl({\cal A},\{p_i\}\bigr)$ is truthful if for every agent $i$, every 
$c_{-i}\in C_{-i}$, and every $\bc_i, c_i\in C_i$, we have 
$u_i(\bc_i,c_{-i};\bc_i)\geq u_i(c_i,c_{-i};\bc_i)$. 
$M$ is IR if for every $i$, every $\bc_i\in C_i$ and every 
$c_{-i}\in C_{-i}$, we have $u_i(\bc_i,c_{-i};\bc_i) \geq 0$. 
\end{definition}

To ensure that truthfulness and IR are compatible, we consider 
{\em monopoly-free} settings: for every player $i$, there is a feasible allocation $a$
(i.e., $\pub(a)<\infty$) with $a_i=\es$. (Otherwise, if there is no such allocation, then
$i$ needs to be paid at least $\min_{v\in\T_i}c_{i,v}$ for IR, so he can lie and increase
his utility arbitrarily.) 

For a {\em randomized mechanism} $M$, where $\cal A$ or the $p_i$'s are randomized,
we say that $M$ is {\it truthful in expectation} if each agent $i$ maximizes his
expected utility by reporting his true cost. We now say that $M$ is 
IR if for {\em every} coin toss of the mechanism, the utility of each agent is nonnegative
upon bidding truthfully.  

Since the CM problem is often \nphard, 
our goal is to design a mechanism 
$M=\bigl(\A,\{p_i\}\bigr)$ that is truthful (or truthful in expectation), and where $\A$
is a $\rho$-approximation algorithm; that is, for every input $c$, the solution $a=\A(c)$
satisfies
$\sum_i c_i(a)+\pub(a)\leq\rho\cdot\min_{b\in\Oc}\bigl(\sum_i c_i(b)+\pub(b)\bigr)$.  
We call such a mechanism a {\em truthful, $\rho$-approximation mechanism}.

The following theorem gives a necessary and sometimes
sufficient condition for when an algorithm $\A$ is {\em implementable}, that is, admits
suitable payment functions $\{p_i\}$ such that $\bigl(\A,\{p_i\}\bigr)$ is a truthful
mechanism. Say that $\cal A$ satisfies {\it weak monotonicity} (WMON) if for all $i$, all
$c_i,c'_i\in C_i$, and all $c_{-i}\in C_{-i}$, 
if ${\cal A}(c_i,c_{-i})=a$, ${\cal A}(c'_i,c_{-i})=b$, then 
$c_i(a)-c_i(b)\leq c'_i(a)-c'_i(b)$. 
Define the dimension of a covering problem to be 
$\max_i|\T_i|$. 
It is easy to see that 
for a single-dimensional covering problem---so $C_i\sse\R$ for all $i$---WMON is
equivalent to the following simpler condition: say that $\A$ is 
{\em monotone} if for all $i$, all $c_i, c'_i\in C_i,\ c_i\leq c'_i$, and all
$c_{-i}\in C_{-i}$, \nolinebreak
\mbox
{if $\A(c_i,c_{-i})=a,\ \A(c'_i,c_{-i})=b$ then $b_i\sse a_i$.}

\begin{theorem}[Theorems 9.29 and 9.36 in~\cite{agt}] \label{wmon} 
If a mechanism $\bigl({\cal A},\{p_i\}\bigr)$ is truthful, then $\cal A$ satisfies WMON.  
Conversely, if the problem is single-dimensional, or if $C_i$ is convex for all $i$, then
every WMON algorithm $\cal A$ is implementable. 
\end{theorem}

\section{A black-box reduction for multidimensional metric \sufl} \label{fl} 
In this section, we consider the multidimensional metric 
{\em uncapacitated facility location} (\ufl) problem and present a 
{\em black-box reduction} from truthful mechanism design to algorithm design. We show
that any $\rho$-approximation algorithm for \ufl satisfying an additional property 
can be converted in a black-box fashion to a
truthful-in-expectation $\rho$-approximation mechanism (Theorem~\ref{flthm}). 
This is the first such result for a multidimensional covering problem. 
As a corollary, we obtain a truthful-in-expectation, 2-approximation mechanism 
(Corollary~\ref{flresult}).

In the mechanism-design version of \ufl, we have a set $\D$ of clients that
need to be serviced by facilities, and a set $\F$ of locations where facilities may be
opened. Each agent $i$ may provide facilities at the locations in $\T_i\sse\F$. By making
multiple copies of a location if necessary, we may assume that the $\T_i$s are
disjoint. Hence, we will simply say ``facility $\ell$'' to refer to the facility at
location $\ell\in\F$. 
For each facility $\ell\in\T_i$ that is opened, $i$ incurs a private opening cost of
$\bff_{i,\ell}$, and assigning client $j$ to an open facility $\ell$ incurs a publicly
known assignment/connection cost $c_{\ell j}$. To simplify notation, given a tuple
$\{f_{i,\ell}\}_{i\in[n],\ell\in\T_i}$ of facility costs, we use $f_\ell$ to denote
$f_{i,\ell}$ for $\ell\in\T_i$.  
The goal is to open a subset $F\sse\F$ of facilities, 
so as to minimize $\sum_{\ell\in F}\bff_{\ell}+\sum_{j\in\D}\min_{\ell\in F}c_{\ell j}$.
We will assume throughout that the $c_{\ell j}$s form a metric.
It will be notationally convenient to allow our algorithms to have the flexibility of
choosing the open facility $\sg(j)$ to which a client $j$ is assigned (instead of 
$\argmin_{\ell\in F}c_{\ell j}$); since assignment costs are public, this does not affect
truthfulness, and any approximation guarantee achieved also clearly holds when we drop
this flexibility.  

\newcommand{\flopt}{\ensuremath{OPT_{\eqref{ufl-p}}}}

We can formulate (metric) \ufl as an integer program, and relax the integrality
constraints to obtain the following LP. Throughout, we use $\ell$ to index facilities in
$\F$ and $j$ to index clients in $\D$. 
\begin{gather}
\min \quad \sum_\ell f_\ell y_\ell+\sum_{j, \ell}c_{\ell j}x_{\ell j}
\qquad \text{s.t.} \qquad
\sum_{\ell} x_{\ell j} \geq 1 \quad \forall j, \qquad
0\leq x_{\ell j} \leq y_\ell \leq 1 \quad \forall \ell, j. \tag{FL-P} \label {ufl-p}
\end{gather}
Here, $\{f_\ell\}_\ell=\{f_{i,\ell}\}_{i\in[n],\ell\in\T_i}$ is the vector of
reported facility costs. Variable $y_\ell$ denotes if facility $\ell$ is opened, and
$x_{\ell j}$ denotes if client $j$ is assigned to facility $\ell$; the constraints encode
that each client is assigned to a facility, and that this facility must be open.

Say that an algorithm $\A$ is a {\em Lagrangian multiplier preserving} (LMP)
$\rho$-approximation algorithm for \ufl if for every instance, it returns a solution
$\bigl(F,\{\sg(j)\}_{j\in\D}\bigr)$ such that 
$\rho\sum_{\ell\in F}f_\ell+\sum_{j}c_{\sg(j) j}\leq\rho\cdot\flopt$.
The main result of this section is the following black-box reduction. 

\begin{theorem}\label{main theorem 3} \label{flthm}
Given a polytime, LMP $\rho$-approximation algorithm $\cal A$ for \ufl, one can construct
a polytime, truthful-in-expectation, individually rational, $\rho$-approximation mechanism
$M$ for multidimensional \ufl. 
\end{theorem}

\begin{proof}
We build upon the convex-decomposition idea used in~\cite{LaviS11}.
The randomized mechanism $M$ works as follows.
Let $f=\{f_\ell\}$ be the vector of reported facility-opening costs, and $c$ be the public
connection-cost metric. 

\begin{list}{{\bf\arabic{enumi}.}}{\usecounter{enumi} \topsep=0.2ex \itemsep=0.5ex
    \addtolength{\leftmargin}{-3ex}}
\item Compute the optimal solution $(y^*,x^*)$ to \eqref{ufl-p} (for the input $(f,c)$). 
Let $\{p^*_i=p^*_i(f)\}$ be the payments made by the {\em fractional} VCG mechanism that 
outputs the optimal LP solution for every input. 
That is, 
$p^*_i=\bigl(\sum_{\ell}f_{\ell}y'_{\ell}+\sum_{\ell,j}c_{\ell j}x'_{\ell j}\bigr)
-\bigl(\sum_{\ell\notin\T_i}f_{\ell}y^*_{\ell}+\sum_{\ell,j}c_{\ell j}x^*_{\ell j}\bigr)$,
where $(y',x')$ is the optimal solution to \eqref{ufl-p} with the additional constraints
$y_\ell=0$ for all $\ell\in\T_i$. 

\item Let $\Z(P)=\{(y^{(q)},x^{(q)})\}_{q\in\I}$ be the set of all integral solutions
to \eqref{ufl-p}. 
In Lemma~\ref{cvdecomp}, we prove the key technical result that using $\A$, one can compute, in
polynomial time, nonnegative multipliers $\{\lambda^{(q)}\}_{q\in\I}$
such that $\sum_q\lambda^{(q)}=1$,  
$\sum_q\lambda^{(q)}y_\ell^{(q)}=y^*_\ell$ for all $\ell$, and
$\sum_{q,\ell,j}\lambda^{(q)}c_{\ell j}x_{\ell j}^{(q)}\leq \rho\sum_{\ell,j}c_{\ell j}x^*_{\ell j}$.

\item With probability $\lambda^{(q)}$: (a) output the solution $\bigl(y^{(q)},x^{(q)}\bigr)$;
(b) pay $p_i^{(q)}$ to agent $i$, where $p_i^{(q)}=0$ if $\sum_{\ell\in\T_i}f_\ell y^*_\ell=0$, and  
$\sum_{\ell\in\T_i}f_\ell y^{(q)}_\ell\cdot\frac{p^*_i}{\sum_{\ell\in\T_i}f_\ell y^*_\ell}$ 
otherwise.  
\end{list}

Clearly, $M$ runs in polynomial time. 
Fix a player $i$. Let $\bff_i$ and $f_i$ be the true and reported cost vector of $i$. Let 
$f_{-i}$ be the reported cost vectors of the other players. 
Let $(y^*,x^*)$ be an optimal solution to \eqref{ufl-p} for $(f,c)$.
Note that $\E{p_i(f)}=p^*_i(f)$
If $\sum_{\ell\in\T_i}f_\ell y^*_\ell=0$ then this follows since $p^*_i(f)=0$ (because
then $(y^*,x^*)$ is also an optimal solution to \eqref{ufl-p} when player $i$ does not
participate). Otherwise, this follows since $\sum_q\ld^{(q)}y^{(q)}=y^*_\ell$ for all
$\ell$.  
So 
$\E{u_i(f_i,f_i;\bff_i)} = \E{p_i}-\sum_q\ld^{(q)}\sum_{\ell\in\T_i}\bff_\ell y^{(q)}_\ell 
= p^*_i(f)-\sum_{\ell\in\T_i}\bff_\ell y^*_\ell$
where the last equality is again because $\sum_q\ld^{(q)}y^{(q)}=y^*_\ell$ for all $\ell$. 
Since $p^*_i$ and $y^*$ are respectively the payment to $i$ and the assignment computed
for input $(f_i,f_{-i})$ by the fractional VCG mechanism, which is truthful, it
follows that player $i$ maximizes his utility in the VCG mechanism, and hence, his
expected utility under mechanism $M$, by reporting his true opening costs. 
Thus, $M$ is truthful in expectation.

This also implies the $\rho$-approximation guarantee because the convex decomposition
obtained in Step 2 shows that the expected cost of the solution computed by $M$ for input
$(f,c)$ (where we may assume that $f$ is the true cost vector)
is at most $\rho\cdot\flopt(f,c)$. 
Finally, since the fractional VCG mechanism is IR, for any agent $i$, the VCG payment
$p^*_i(f)$ satisfies $p^*_i(f)\geq \sum_{\ell\in\T_i}f_\ell y^*_\ell$, and therefore 
$p_i^{(q)}\geq\sum_{\ell\in\T_i}f_\ell y^{(q)}_\ell$. \nolinebreak
{So $M$ is IR.}
\end{proof}

\begin{lemma} \label{cvdecomp}
The convex decomposition in step 2 can be computed in polytime.
\end{lemma}

\begin{proof}
It suffices to show that the LP \eqref{primal} can be solved in polynomial time 
and its optimal value is $1$. Recall that $\{(y^{(q)},x^{(q)})\}_{q\in\I}$ is the set of
all integral solutions to \eqref{ufl-p}. 

\vspace{-3ex}

{\centering\small
\noindent \hspace*{-6ex}
\begin{minipage}[t]{.475\textwidth}
\begin{alignat}{2}
\max & \quad & \sum_q\lambda^{(q)} & \tag{P} \label{primal} \\
\text{s.t.} && \sum_q\lambda^{(q)} y_\ell^{(q)} & = y_\ell^* \qquad \forall \ell \label{1} \\
&& \sum_{j,\ell,q}\lambda^{(q)}c_{\ell j}x_{\ell j}^{(q)} & 
\leq \rho\sum_{j,\ell}c_{\ell j}x_{\ell j}^* \label{2} \\
&& \sum_q\ld^{(q)} & \leq 1 \label{3} \\ 
&& \lambda & \geq 0. \notag
\end{alignat}
\end{minipage}
\quad \rule[-31ex]{1pt}{28ex}\ 
\begin{minipage}[t]{.52\textwidth}
\begin{alignat}{2}
\min & \quad & \sum_\ell y_\ell^*\alpha_\ell 
+\bigl(\rho\sum_{j,\ell}c_{\ell j}x_{\ell j}^*\bigr)\beta & + z \tag{D} \label{dual} \\ 
\text{s.t.} && \sum_\ell y_\ell^{(q)}\alpha_\ell
+\bigl(\sum_{j,\ell}c_{\ell j}x_{\ell j}^{(q)}\bigr)\beta + z & \geq 1 \qquad \forall q 
\label{4} \\
&& z, \beta & \geq 0. \notag
\end{alignat}
\end{minipage}
}

\smallskip
Since \eqref{primal} has an exponential number of variables, we consider the dual 
\eqref{dual}.  
Here the $\alpha_\ell$s, $\beta$ and $z$ are the dual variables corresponding to  
constraints \eqref{1}, \eqref{2}, and \eqref{3} respectively.
Clearly, $\OPT_{\eqref{dual}}\leq 1$ since $z=1$, $\al_\ell=0=\beta$ for all
$\ell$ is a feasible dual solution. If there is a feasible dual solution
$(\alpha',\beta',z')$ of value smaller than 1, then the rough idea is that by running
$\cal A$ on the \ufl instance with facility costs $\{\frac{\alpha'_\ell}{\rho}\}$ and
connection costs $\{\beta' c_{\ell j}\}$, we can obtain an integral solution whose
constraint \eqref{4} is violated. (This idea needs be modified a bit since $\al'_\ell$
could be negative; see below.) 
Hence, we can solve \eqref{dual} efficiently via the ellipsoid method using $\A$ to
provide the separation oracle. This also yields an equivalent dual LP consisting of only
the polynomially many violated inequalities found during the ellipsoid method. The dual of
this compact LP gives an LP equivalent to \eqref{primal} with polynomially many
{$\ld^{(q)}$ variables whose solution yields the desired convex decomposition.}

We now fill in the details.
Suppose $(\alpha',\beta',z')$ is feasible to \eqref{dual} and
$\sum_\ell y_\ell^*\alpha'_\ell+(\rho \sum_{j,\ell}c_{\ell j}x_{\ell j}^*)\beta'+z'<1$.
Define $a^+:=\max(0,a)$; for a vector $v=(v_1,\ldots,v_n)$, define
$v^+:=(v_1^+,\ldots,v_n^+)$. 
Consider the \ufl instance with facility costs $\{f'_\ell=\alpha'^+_\ell/\rho\}$ and
connection costs $\{c'_{\ell j}=\beta' c_{\ell j}\}$. (Clearly $c'$ is also a metric.)
Running $\cal A$ on this input, we can obtain an integral solution $(y^{(q)},x^{(q)})$
such that  
$$
\rho\sum_\ell\tfrac{\alpha_\ell^{'+}}{\rho}y_\ell^{(q)}
+\sum_{j,\ell}\beta'c_{\ell j}x_{\ell j}^{(q)}
\leq \rho\cdot\flopt(f',c')
\leq \rho\Bigl(\sum_\ell\tfrac{\alpha_\ell^{'+}}{\rho}y_\ell^*
+\sum_{j,\ell}\beta'c_{\ell j}x_{\ell j}^*\Bigr).
$$
Clearly the facilities $\ell$ with $\al'_\ell\leq 0$ contribute 0 to the LHS and RHS of
the above inequality.
Now consider the integer solution $\hy^{(q)}$ where $\hat{y}_\ell^{(q)}$ is 1 if
$\alpha_\ell'\leq 0$ and is $y^*_\ell$ otherwise.  
Adding $\sum_{\ell:\al'_\ell\leq 0}\al'_\ell\hy^{(q)}_\ell$ to the LHS and
$\sum_{\ell:\al'_\ell\leq 0}\al'_\ell y^*_\ell$ to the RHS of the above inequality, 
since $y^*_\ell\leq 1$ for all $\ell$ and $\al_\ell^{'+}=\al'_\ell$ when $\al'_\ell>0$, we
infer that  
$$
\sum_\ell\alpha'_\ell\hy_\ell^{(q)}+\sum_{j,\ell}\beta'c_{\ell j}x_{\ell j}^{(q)}
\leq \sum_\ell\alpha'_\ell y_\ell^*+\bigl(\rho\sum_{j,\ell}c_{\ell j}x_{\ell j}^*\bigr)\beta'
<1-z'
$$
which contradicts that $(\alpha',\beta',z')$ is feasible to \eqref{dual}.
Hence, $\OPT_{\eqref{dual}}=\OPT_{\eqref{primal}}=1$.

Thus, we can add the constraint 
$\sum_\ell y_\ell^*\alpha_\ell+(\rho\sum_{j,\ell}c_{\ell j}x_{\ell j}^*)\beta + z\leq 1$ to
\eqref{dual} without altering anything. If we solve the resulting LP using the ellipsoid
method, and take the inequalities corresponding to the violated inequalities \eqref{4}
found by $\A$ during the ellipsoid method, then we obtain a compact LP with only a
polynomial number of constraints that is equivalent to \eqref{dual}. The dual of this
compact LP yields an LP equivalent to \eqref{primal} with a polynomial number of
$\ld^{(q)}$ variables which we can solve to obtain the desired convex decomposition. 
\end{proof}

By using the polytime LMP 2-approximation algorithm for \ufl devised by Jain et
al.~\cite{Jain}, we obtain the following corollary of Theorem~\ref{flthm}.

\begin{theorem} \label{flresult}
There is a polytime, IR, truthful-in-expectation, 2-approximation mechanism for
multidimensional \ufl. 
\end{theorem}

\section{Truthful mechanisms for multidimensional \svcp} \label{vc} 
We now consider the multidimensional vertex-cover problem (\vcp), and devise various  
polytime, truthful, approximation mechanisms for it. We often use \mvcp to distinguish 
multidimensional \vcp from its algorithmic counterpart.  

Recall that in \mvcp, we have a graph $G=(V,E)$ with $n$ nodes. Each agent $i$
provides a subset $\T_i$ of nodes. For simplicity, we first assume that the $\T_i$s are
disjoint, and given a cost-vector $\{c_{i,u}\}_{i\in [n],u\in\T_i}$, we use $c_u$ to
denote $c_{i,u}$ for $u\in\T_i$. 
Notice that monopoly-free then means that each $\T_i$ is an independent set. 
In Remark~\ref{genvcp} we argue that many of the results obtained in this 
disjoint-$\T_i$s setting (in particular, Theorems~\ref{rdim} and~\ref{mclosed}) also hold
when the $\T_i$s are not disjoint (but each $\T_i$ is still an independent set).  
The goal is to choose a minimum-cost {\em vertex cover}, i.e., a min-cost set $S\sse V$
such that every edge is incident to a node in $S$.

As mentioned earlier, 
\vcp becomes a rather challenging mechanism-design problem in the {\em multidimensional}  
mechanism-design setting. 
Whereas for {\em single-dimensional \vcp}, many of the known 2-approximation algorithms
for \vcp are implementable,
none of these underlying techniques yield implementable algorithms even for the simplest
multidimensional setting, 2-dimensional \vcp, 
where {\em every player owns at most two nodes}; see Appendix~\ref{lpround-bad}
and~\ref{other-bad} for examples.
Moreover, no maximal-in-distributional-range (MIDR) mechanism whose range is a
proper subset of all outcomes can achieve a bounded multiplicative approximation
guarantee~\cite{BBShaddin}.%
\footnote{If $\A$ is a randomized MIDR algorithm and $S$ is an inclusion-wise minimal 
vertex cover such that the range of $\A$ does not include a distribution that
returns $S$ with probability 1, then $\A$ incurs non-zero cost on the instance where the
cost of a node $u$ is 0 if $u\in S$ and is 1 (say) otherwise, and so its approximation
ratio is unbounded.}
This also rules out the convex-decomposition technique of~\cite{LaviS11}, 
which yields MIDR mechanisms.
 
We develop two main techniques for \mvcp in this section. In
Section~\ref{thresh}, we introduce a simple class of truthful mechanisms called 
{\em threshold mechanisms}, and show that although seemingly restricted,
threshold mechanisms can achieve non-trivial 
approximation guarantees. In Section~\ref{decomp}, we develop a {\em decomposition method}
for \mvcp that uses threshold mechanisms as building blocks and gives a general way of
reducing the mechanism-design problem for \mvcp into simpler
mechanism-design problems. 

By leveraging the decomposition method along with threshold mechanisms, we obtain various
truthful, approximation mechanisms for \mvcp, which yield the {\em first} truthful
mechanisms for multidimensional vertex cover with non-trivial approximation guarantees. 
(1) We show that any instance of $r$-dimensional \vcp can be decomposed into 
$O(r^2\log n)$ {\em single-dimensional \vcp} instances; this leads to a 
truthful, $O(r^2\log n)$-approximation mechanism for $r$-dimensional \vcp
(Theorem~\ref{rdim}). In particular, for any fixed $r$, we obtain an 
$O(\log n)$-approximation. 
(2) For any proper minor-closed family of graphs (such as planar graphs), we obtain an 
improved truthful, $O(r\log n)$-approximation mechanism (Theorem~\ref{mclosed}); this
improves to an $O(\log n)$-approximation if no two neighbors of a node belong to the same
agent (Corollary~\ref{mclosed-3hop}).  

Theorem~\ref{frugality} shows that our mechanisms also enjoy good frugality properties.
We obtain the first mechanisms for \mvcp that are polytime, truthful, and
achieve bounded approximation ratio {\em and} bounded frugality ratio. 
This nicely complements a result of~\cite{ElkindFOCS}, 
{who devise such mechanisms for single-dimensional \vcp.}

\subsection{Threshold Mechanisms} \label{edge-threshold} \label{thresh}
\begin{definition} \label{threshold mechanism} \label{threshmech}
A {\em threshold mechanism} $M$ for \mvcp works as follows. On input $c$, for every $i$
and every node $u\in\T_i$, $M$ computes a threshold $t_u=t_u(c_{-i})$ (i.e., $t_u$
does not depend on $i$'s reported costs). $M$ then returns the solution 
$S=\{v\in V: c_v\leq t_v\}$ as the output, and pays $p_i=\sum_{u\in S\cap\T_i} t_u$ to
agent $i$. 
\end{definition}

\noindent
If $t_u$ only depends on the costs in the neighbor-set $N(u)$ of $u$, for all
$u\in V$ (note that $N(u)\cap\T_i=\es$ if $u\in\T_i$), we call $M$ a 
{\em neighbor-threshold mechanism}. A special case of a neighbor-threshold 
mechanism is an {\em edge-threshold mechanism}: for every edge $uv\in E$ we have edge
thresholds $t_u^{(uv)}=t_u^{(uv)}(c_v)$,  $t_v^{(uv)}=t_v^{(uv)}(c_u)$, and
the threshold of a node $u$ is given by $t_u=\max_{v\in N(u)}(t_u^{(uv)})$. 

In general, threshold mechanisms may not output a vertex cover, however it is easy to
argue that threshold mechanisms are always truthful and IR.

\begin{lemma} \label{threshold mech is truthful} \label{threshtruth}
Every threshold mechanism 
for \mvcp is IR and truthful. 
\end{lemma}

\begin{proof}
IR is immediate from the definition of payments.
To see truthfulness, fix an agent $i$. For every $\bc_i,c_i\in C_i, c_{-i}\in C_{-i}$ we have
$u_i(c_i,c_{-i};\bc_i)=\sum_{v\in\T_i: c_v\leq t_v}(t_v-\bc_v)$.
It follows that $i$'s utility is maximized by reporting $c_i=\bc_i$.
\end{proof}

Inspired by~\cite{KempeFOCS,ElkindFOCS}, we define an {\em $x$-scaled} edge-threshold
mechanism as follows: fix a vector $(x_u)_{u\in V}$, where $x_u>0$ for all $u$, and 
set $t_u^{(uv)}:=x_uc_v/x_v$ for every edge $(u,v)$. We abuse notation and use 
${\cal A}_x$ to denote both the resulting edge-threshold mechanism and its allocation
algorithm. Also, define $\B_x$ to be the neighbor-threshold mechanism where we set 
$t_u:=\sum_{v\in N(u)}x_uc_v/x_v$. 
Define $\alpha(G; x):=
\max_{u\in V}\bigl(\max_{S\subseteq N(u): S\text{ independent}}\frac{x(S)}{x_u}\bigr)$.

\begin{lemma}
\label{x-scaled theorem} 
${\cal A}_x$ and $\B_x$ output feasible solutions and have a tight 
approximation ratio $\alpha(G;x)+1$.  
\end{lemma}

\begin{proof}
Clearly, every node selected by ${\cal A}_x$ is also selected by $\B_x$. So it suffices to
show that $\A_x$ is feasible, and to show the approximation ratio for $\B_x$.
For any edge $(u,v)$, either $c_u\leq x_uc_v/x_v$ and $u$ is output, 
or $c_v \leq x_vc_u/x_u$ and $v$ is output. 
So $\A_x$ returns a vertex cover. 

Let $S$ be the output of $\B_x$ on input $c$, and let $S^*$ be a min-cost vertex
cover. We have 
$c(S)=c(S\cap S^*)+c(S\setminus S^*)\leq c(S^*)+\sum_{u\in S\setminus S^*}t_u
=c(S^*)+\sum_{u\in S\sm S^*}\sum_{v\in N(u)}x_uc_v/x_v$.
Note that $S\sm S^*$ is an independent set since $S^*$ is a vertex cover, so 
$\sum_{u\in S\sm S^*}\sum_{v\in N(u)}x_uc_v/x_v
\leq\sum_{v\in  S^*}\frac{c_v}{x_v}\sum_{u\in N(v)cap S^*}x_u
\leq\sum_{v\in  S^*}c_v\cdot\al(G;x)$. 
Hence $c(S)\leq(\alpha(G;x)+1)c(S^*)$.
The tightness of the approximation guarantee follows from Example~\ref{xscale} below.
\end{proof}

\begin{corollary} \label{xscmech}
(i) Setting
$x=\vec{1}$ gives $\alpha(G;x)\leq\Delta(G)$, 
which is the maximum degree of a node in $G$, so ${\cal A}_{\vec{1}}$ has approximation
ratio at most $\Delta(G)+1$. %

\noindent
(ii) Taking $x$ to be the eigenvector corresponding to the largest eigenvalue $\ld_{\max}$
of the adjacency matrix of $G$ ($x>0$ by the Perron-Frobenius theorem) gives
$\alpha(G;x)\leq\lambda_{max}$ (see~\cite{ElkindFOCS}), so ${\cal A}_{x}$ has
approximation ratio $\lambda_{max}+1$. 
\end{corollary}

\begin{example}[Tightness of approximation ratio of $\A_x$ and $\B_x$] \label{xscale}
Let $u$ and $S\subseteq N(u)$ achieve the maximum in the definition of $\alpha(G;x)$. 
Now consider the instance $(G,c)$ where $c_u=x_u$, $c_v=x_v$ for all $v\in S$ and $c_w=0$
for all $w\in V\setminus (\{u\}\cup S)$. 
The mechanism ${\cal A}_x$ will choose $\{u\}\cup S$ in the output, whereas $V\sm S$ is a 
vertex cover of cost $c_u=x_u$. So, ${\cal A}_x$ has approximation ratio at least
$\frac{x_u+x(S)}{x_u}=1+\alpha(G;x)$. 
\end{example}

Although neighbor-threshold mechanisms are more general than
edge-threshold mechanisms, Lemma~\ref{nbredgemap} (proved in Appendix~\ref{append-vc})
shows that this yields limited dividends in the approximation ratio.
Define 
$\al'(G)=\min_{\text{orientations of $G$}}
\bigl(\max_{u\in V, S\sse N^{\into}(u): S\text{ independent}}|S|\bigr)$, where 
$N^\into(u)=\{v\in N(u): (u,v)\text{ is directed into }u\}$. 
Note that $\al'(G)\leq\al(G;\vec{1})\leq\Dt(G)$. 
If $G=(V,E)$ is {\em everywhere $\gm$-sparse}, i.e., 
$|\{(u,v)\in E: u,v\in S\}|\leq \gm|S|$ for all $S\sse V$, then $\al'(G)\leq \gm$; this
follows from Hakimi's theorem~\cite{Hakimi}. 
A well-known result in graph theory states that for every proper family $\G$ of graphs
that is closed under taking minors (e.g., planar graphs), there is a constant $\gm$, such
that every $G\in\G$ is has at most $\gm|V(G)|$ edges~\cite{Mader67} (see also~\cite{Diestel},
Chapter 7, Exer. 20); since $\G$ is minor-closed, this also implies that $G$ is 
{\em everywhere} $\gm$-sparse, and hence $\al'(G)\leq\gm$ for all $G\in\G$. 

\begin{lemma} \label{nbredgemap}
A (feasible) neighbor-threshold mechanism $M$ for graph $G$ with approximation ratio 
$\rho$, yields an $O\bigl(\rho \log(\alpha'(G))\bigr)$-approximation edge-threshold
mechanism for $G$. This implies an approximation ratio of
(i) $O(\rho\log\gm)$ if $G$ is an everywhere $\gm$-sparse graph; 
(ii) $O(\rho)$ if $G$ belongs to a proper minor-closed family of graphs (where the
constant in the $O(.)$ depends on the graph family).
\end{lemma}

\begin{remark} \label{genvcp}
Any neighbor-threshold mechanism $M$ with approximation ratio $\rho$ that works 
under the disjoint-$\T_i$s assumption can be modified to yield a truthful,
$\rho$-approximation mechanism when we drop this assumption. 
Let $A_u=\{i: u\in\T_i\}$.
Set $\hat{c}_u=\min_{i\in A_u} c_{i,u}$ for each $u\in V$ and let $\hht_u$ be the
neighbor-threshold of $u$ for the input $\hc$. 
Note that $\hht_u$ depends only on $c_{-i}$ for every $i\in A_u$.
Set $t_u^i:=\min\{\hat{t}_u, \min_{j\neq i: u\in\T_j} c_{j,u}\}$ for all $i, u\in\T_i$.
Consider the threshold mechanism $M'$ with $\{t_u^i\}$ thresholds, where we use a fixed  
tie-breaking rule to ensure that we pick $u$ for at most one agent $i\in A_u$ with
$c_{i,u}=t_u^i$. Then the outputs of $M$ on $c$, and of $M'$ on input $\hc$ coincide. 
{Thus, $M'$ is a truthful, $\rho$-approximation mechanism.}
\end{remark}

\subsection{A decomposition method} \label{decomp} 
We now propose a general reduction method for \mvcp that uses threshold mechanisms as
building blocks to reduce the task of designing truthful mechanisms for \mvcp to the task
of designing threshold mechanisms for simpler (in terms of graph structure or the
dimensionality of the problem) \mvcp problems. This reduction is useful because designing
good threshold mechanisms appears to be a much more tractable task for \mvcp. 
By utilizing the threshold mechanisms designed in Section~\ref{thresh} in our
decomposition method, we obtain an $O(\log n)$-approximation mechanism for any proper 
minor-closed family of graphs, 
and an $O(r^2\log n)$-approximation mechanism for $r$-dimensional \vcp. 

A {\it decomposition mechanism} $M$ for $G=(V,E)$ is constructed as follows. 
\begin{list}{--}{\topsep=0ex \itemsep=0ex \addtolength{\leftmargin}{-1ex}}
\item Let $G_1,\ldots,G_k$ be subgraphs of $G$ such that $\bigcup_{q=1}^k E(G_q)=E$,
\item Let $M_1,\ldots,M_k$ be threshold mechanisms for $G_1,\ldots,G_k$ respectively. 
For any $v\in V$, let $t_v^q$ be $v$'s threshold in $M_q$ if $v\in V(G_i)$, and $0$
otherwise. 
\item Define $M$ to be the threshold mechanism obtained by setting the threshold for
each node $v$ to $t_v:=\max_{q=1,\ldots,k}(t_v^q)$ for any $v\in V$. The payments of $M$
are then as specified in Definition~\ref{threshmech}. Notice that if all the $M_i$s are
neighbor threshold mechanisms, then so is $M$.
\end{list}

\begin{lemma} \label{decomp theorem}
The decomposition mechanism $M$ described above is IR and truthful.
If $\rho_1,\ldots,\rho_k$ are the approximation ratios of $M_1,\ldots,M_k$
respectively, then $M$ has approximation ratio $\bigl(\sum_q\rho_q\bigr)$.
\end{lemma}

\begin{proof}
Since $M$ is a threshold mechanism, it is IR and truthful by Lemma~\ref{threshtruth}. 
The optimal vertex cover for $G$ induces a vertex cover for each subgraph $G_q$. So $M_q$
outputs a vertex cover $S_q$ of cost at most $\rho_q\cdot\OPT$, where $\OPT$ is the
optimal vertex-cover cost for $G$. It is clear that $M$ outputs $\bigcup_q S_q$, which has
cost at most $\bigl(\sum_q\rho_q\bigr)\cdot\OPT$. 
\end{proof}

\begin{theorem} \label{1 dim decomp} \label{rdim}
For any $r$-dimensional instance of \mvcp on $G=(V,E)$, one can obtain a  
polytime, $O(r^2\log|V|)$-approximation, decomposition mechanism, even when the $\T_i$s 
are not disjoint. 
\end{theorem}

\begin{proof}
We decompose $G$ into single-dimensional subgraphs, by which we mean subgraphs that
contain at most one node from each $\T_i$.
Initialize $j=1$, $V_j=\es$. While, $\bigcup_{q=1}^{j-1} E(G_q)\neq E$, we do the
following: for every agent $i$, we pick one of the nodes of $\T_i$ uniformly at random and
add it to $V_j$. We also add all the nodes in $V\sm\bigl(\bigcup_{i=1}^n\T_i\bigr)$ to
$V_j$. Let $G_j$ be the induced subgraph on $V_j$; set $j\assign j+1$. 

For any edge $e=(u,v)\in E$, the probability that both $u, v$ appear in some subgraph
$G_j$ is at least $1/r^2$. 
So, the expected value of $|E\sm \bigcup_{q=1}^{j-1}E(G_q)|$ decreases by a factor of at
least $(1-1/r^2)$ with $j$. Hence, the expected number of subgraphs produced
above is $O\bigl(\frac{\log |E|}{\log (r^2/(r^2-1))}\bigr)=O(r^2\log |V|)$ (this also
holds with high probability). 
Each $G_j$ yields a single-dimensional \vcp instance (where a node may be owned by
multiple players). Any truthful mechanism for a 1D-problem is a threshold mechanism. So we
can use any truthful, 2-approximation mechanism for single-dimensional \vcp for the $G_j$s
and obtain an $O(r^2\log n)$-approximation for $r$-dimensional \vcp.  
\end{proof}

The following lemma shows that the decomposition obtained above into single-dimensional
subgraphs is essentially the best that can hope for, for $r=2$. 

\begin{lemma}
There are instances of $2$-dimensional \vcp that require $\Omega(\log|V(G)|)$
single-dimensional subgraphs in any decomposition of $G$.
\end{lemma}

\begin{proof}
Define $G^n$ to be the bipartite graph with vertices $\{u_1,\ldots,u_n,v_1,\ldots,v_n\}$
and edges $\{(u_i,v_j): i\neq j\}$. Each agent $i=1,\ldots,n$ owns vertices $u_i$
and $v_i$. 

For $n=2$ the claim is obvious. Let $q_n$ be the minimum number of single-dimensional 
subgraphs needed to decompose $G^n$. Suppose the claim is true 
for all $j<n$ and we have decomposed $G^n$ into single-dimensional 
subgraphs $D=\{G_1,\ldots,G_{q_n}\}$. 
We may assume that $V(G_1)=\{u_1,\ldots,u_k,v_{k+1},\ldots,v_n\}$ (if $G_1$ has less  
than $n$ nodes, pad it with extra nodes). 
Let $H_1$ and $H_2$ be the subgraphs of $G$ induced by 
$\{u_1,\ldots,u_k,v_1,\ldots,v_k\}$ and $\{u_{k+1},\ldots,u_n,v_{k+1},\ldots,v_n\}$,  
respectively. The graphs in $D\setminus \{G_1\}$ must contain a decomposition of $H_1$ and 
a decomposition of $H_2$. So $q_n\geq 1+\max(q_{k},q_{n-k})$, and hence, by induction,
we obtain that $q_n\geq 1+(1+\log_2(n/2))=1+\log_2 n$. 
\end{proof}

Complementing Theorem~\ref{rdim}, 
we next present another decomposition mechanism that exploits the graph structure to
obtain an improved approximation guarantee. Given a graph $G=(V,E)$ and a set $S\sse V$,
we use $E[S]$ to denote the set of edges having both end points in $S$, and 
$N(S)=\{u\in V\sm S: \exists v\in S\text{ s.t. }(u,v)\in E\}$ to denote the neighbors of
$S$. Also, let $\dt(S,T)$ denote the set of edges of $G$ having one end point each in $S$
and $T$.
When we subscript a quantity (e.g., $\dt(S)$ or $N(S)$) with a specific graph, we are 
referring to the quantity in that specific graph. 

\begin{theorem} \label{sparse theorem} \label{mclosed}
If $G=(V,E)$ is everywhere $\gm$-sparse, then one can devise a polytime,
$O(\gm r\log |V|)$-approximation decomposition mechanism for $r$-dimensional \vcp on
$G$. Hence, there is a polytime, truthful, $O(r\log n)$-approximation mechanism for
$r$-dimensional \vcp on any proper minor-closed family of graphs. These guarantees also
hold when the $\T_i$s are not disjoint.  
\end{theorem}

\begin{proof}
Set $G=G_0=(V_0,E_0)$, and let $n_0=|V_0|$.
Since $|E_0|\leq\gm n_0$, there are at most $n_0/2$ nodes in $V_0$ with degree larger than
$4\gm$. Let $T_1=\{u\in V_0: \dt(u)\leq 4\gm\}$.
Let $H_1=\bigl(T_1,E[T_1]\bigr)$ be the subgraph of $G_0$ induced by $T_1$.
Also, consider the bipartite subgraph 
$B_1=\bigl(T_1\cup N_{G_0}(T_1),\dt_{G_0}(T_1,N_{G_0}(T_1))\bigr)$. 
Now, $G_1=G_0\sm T_1$ (i.e., we delete the nodes in $T_1$ and the edges incident to them
to obtain $G_1$) is also $\gm$-sparse. So, we can similarly find a subgraph $H_2$ 
that contains at least half of the nodes of $G_1$, and the bipartite subgraph $B_2$ of
$G_1$. Continuing this process, we obtain subgraphs $H_1,B_1,H_2,B_2,\ldots,H_k,B_k$ that
partition $G$, where for every $q$, each node of $H_q$ and each node on one of the sides
of $B_q$ has degree (in that subgraph) at most $4\gm$, and 
$|V(H_q)|\geq|V(G\sm (T_1\cup\ldots T_{q-1})|/2$. Hence, $k\leq\log n$. 
Using the (edge-threshold) mechanism $\A_{\vec{1}}$ defined in Corollary~\ref{xscmech},
for each $H_q$ subgraph gives a $(4\gm+1)$-approximation for each $H_q$.
Let $B_q=\bigl(T_q\cup R_q,F_q)$, where $R_q=N_{G_{q-1}}(T_q)$, and
$F_q=\dt_{G_{q-1}}(T_q,R_q)$. 

Let $T=\bigcup_q T_q$, $R=\bigcup_q R_q$. 
Note that a node $u$ could lie in $T\cap R$. We replace each such node $u\in T\cap R$ with
two distinct ``copies'' $u_1$ and $u_2$, and place $u_1$ in $T$ and $u_2$ in $R$. 
If $u\in\T_i$ for some player $i$, then we include both $u_1, u_2$ in $\T_i$, and set
$c_{i,u_1}=c_{i,u_2}=c_{i,u}$.   
The understanding is that if any of $u_1$ or $u_2$ is picked, then we pick $u$; in other
words, the threshold of $u$ is the maximum of the thresholds of $u_1$ and $u_2$.  
Let $T\uplus R$ denote the resulting set of nodes (with bipartition $T, R$).
We create a bipartite graph $B=(T\uplus R, F)$ representing the union of all the $B_q$s,
where $F$ is defined as follows. For notational simplicity, if a node $u$ is in exactly
one of $T$ and $R$ (so it has only one copy in $T\uplus R$), we set $u_1=u_2=u$.
For every $q=1,\ldots,k$, and every edge $(u,v)\in F_q$, where $u\in R_q$, $v\in T_q$, we
include the edge $(u_2,v_1)$ in $F$.
Note that:
(a) $B$ is bipartite; 
(b) the maximum degree of $T$ (in $B$) is at most $4\gm$; and, 
(c) every edge in $E\sm\bigcup_q E(H_q)$ maps to exactly one edge of $F$. 
We show that one can obtain an $O(r\gm\log n)$-approximation decomposition mechanism for
$B$. Thus, we obtain an $O(r\gm\log n)$-approximation decomposition mechanism for $G$.

We obtain $O(r\log n)$ bipartite graphs whose edges cover $F$, with the property that in
each resulting bipartite subgraph $Z$, for each node $u\in R\cap V(Z)$, and each agent
$i$, {\em at most one} of $u$'s neighbors in $Z$ is in $\T_i$. 
We use a procedure similar to that in the proof of Theorem~\ref{rdim}.  
For each $i$, we pick one node from $T\cap \T_i$ uniformly at random; let $X$ be the set
of nodes picked from $T$. We create the bipartite graph $Z^j$ consisting of all edges
between $X$ and $N_B(X)$. We increment $j$ and continue this process until all edges of
$F$ have been covered. Since the probability that an edge $(u,v)\in F$ is covered in an
iteration is at least $\frac{1}{r}$, $O(r\log n)$ subgraphs suffice, in expectation and
with high probability, to cover $F$.  

Now, for each bipartite graph $Z^j$ with bipartition $X^j\cup Y^j$, where 
$X^j\sse T,\ Y^j\sse R$, we use the following threshold mechanism.  
Assume for now that the $\T_i$s are disjoint, and set $c_u=c_{i,u}$ if $u\in\T_i$.
For each $u\in Y^j$, we pick $u$ if $c_u\leq\sum_{v\in N_{Z^j}(u)}c_v$, and we pick
$N_{Z^j}(u)$ if $\sum_{v\in N_{Z^j}(u)}c_v\leq c_u$. Note that since 
$|X^j\cap\T_i|\leq 1$ for every $i$, 
this is a valid threshold mechanism. 
The cost of the solution $S$ output by this mechanism for $Z^j$ is at most 
$2\sum_{u\in Y^j}c(S^*_u)$, where $S^*_u$ is the optimal vertex cover for the star
consisting of $u$ and $N_{Z^j}(u)$. Since every node in $X^j$ has degree at most $4\gm$,
it is not hard to see that $\sum_{u\in Y^j}c(S^*_u)\leq 4\gm\cdot\OPT(Z^j)$, where
$\OPT(Z^j)$ is the value of an optimal vertex cover for $Z^j$. This follows since, for   
example, concatenating the optimal dual solutions corresponding to the $S^*_u$s and 
scaling by $4\gm$ yields a feasible solution to the dual of the vertex-cover LP for
$Z^j$. Therefore, the threshold mechanism for $Z^j$ is an $8\gm$-approximation, and hence
we obtain an $O(r\gm\log n)$-approximation for $B$. 

\medskip
If the $\T_i$s are not disjoint, then by Remark~\ref{genvcp}, the $O(\gm)$-approximation
for the $H_q$s still holds. When constructing $Z^j$, we set the ``owners'' of a node
$v\in T$ included in $Z^j$ to be all the agents $i$ who picked $v$ as the random node
from their $\T_i$-set (and hence caused $v$ to be included in $Z^j$); the owners of a node
$u\in Y^j$ are unchanged, that is, $\{i: u\in\T_i\}$.
Now, as in Remark~\ref{genvcp}, we can move from this to an instance where each node is
owned by at most one agent. Although the mechanism described above for $Z^j$ is 
{\em not} a neighbor-threshold mechanism, it is not hard to see that since the threshold
for a node $v\in T\cap V(Z^j)$ depends only on nodes that are at hop-distance at most
2 from $v$, none of which are owned by any agent owning $v$ in $Z^j$, the same reasoning
as in Remark~\ref{genvcp} shows that the $O(\gm)$-approximation threshold mechanism
obtained above for $Z^j$ holds even when a node is owned by multiple agents. Thus, we
still obtain an $O(\gm r\log |V|)$-approximation mechanism.

\medskip
As noted in Section~\ref{thresh}, every proper minor-closed family of graphs is everywhere
$\gm$-sparse for some $\gm>0$. Thus, the above result implies a truthful, 
$O(r\log^2 n)$-approximation for any proper minor-closed family (where the constant in the 
$O(.)$ depends on the graph family; e.g., for planar graphs $\gm\leq 4$).
\end{proof}

Given a graph $G=(V,E)$, define a {\em 3-hop-far} instance of \mvcp on $G$ to be one that
satisfies $|N(u)\cap\T_i|\leq 1$ for every $u\in V$ and every agent $i$; 
that is no two neighbors of a node are owned by the same agent. On such instances, one can
improve the guarantee of Theorem~\ref{mclosed} by removing the dependence on
$\max_i|\T_i|$. 

\begin{corollary} \label{mclosed-3hop}
Let $G=(V,E)$ be an everywhere $\gm$-sparse graph. One can devise a polytime 
$O(\gm\log |V|)$-approximation decomposition mechanism for 3-hop-far instances of \mvcp on
$G$. Hence, one obtains a polytime, truthful $O(\log n)$-approximation mechanism for 3-hop
far \mvcp on any proper minor-closed family of graphs. These guarantees also hold when the
$\T_i$s are not disjoint.
\end{corollary}

\begin{proof}
The proof follows from that of Theorem~\ref{mclosed}. The only change is that we no longer
need to decompose the bipartite graph $B$ into the $Z^j$ subgraphs: since the input is a
3-hop-far \mvcp instance, the \mvcp instance on $B$ already satisfies the property
required of the $Z^j$ graphs. 
Thus, we obtain an $O(\gm)$-approximation for $B$, and an $O(\gm)$-approximation
for each $H_q$, and hence an $O(\gm\log |V|)$-approximation for $G$. The consequences when 
the $\T_i$s are not necessarily disjoint, and for a proper minor-closed family of graphs
follow as in the proof of Theorem~\ref{mclosed}.
\end{proof}

\paragraph{Frugality considerations.}
Karlin et al.~\cite{Karlin} and Elkind et al.~\cite{Elkind1} propose various benchmarks
for measuring the {\em frugality ratio} of a mechanism, which is a measure of the
(over-)payment of a mechanism. The mechanisms that we devise above also enjoy good
frugality ratios with respect to the following benchmark introduced by~\cite{Elkind1},
which is denoted by $\nu(G,c)$ in~\cite{KempeFOCS} (and NTU$_{\max}$ in~\cite{Elkind1}).  

\begin{definition}[Frugality benchmark \boldmath $\nu(G,c)$~\cite{Karlin,Elkind1}] 
\label{frbench} 
Given an instance of \vcp on a graph $G=(V,E)$ with node costs $\{c_u\}$, we define
$\nu(G,c)$ as follows.
Fix an arbitrary min-cost vertex cover $S$ (with respect to $c$).%
\footnote{Elkind et al.~\cite{Elkind1} prove that $\nu(G,c)$ does not depend on the
  specific min-cost vertex cover $S$ used in the definition.}   
\begin{alignat*}{3}
\nu(G,c)\ :=\ \max & \quad & \sum_{v\in S}x_v & \\
\text{s.t.} && x_v & \geq c_v \qquad && \frall v\in S \\
&& \sum_{v\in S\setminus T}x_v & \leq \sum_{v\in T\setminus S}c_v \qquad && 
\frall\ \text{vertex covers $T$.}
\end{alignat*}
\end{definition}

The {\it frugality ratio} of a mechanism $M=\bigl(\A,\{p_i\}\bigr)$ on $G$
is defined as $\phi_M(G):=\sup_c\frac{\sum_i p_i(c)}{\nu(G,c)}.$
The proof of Lemma~\ref{x-scaled theorem} is easily modified to show that the $x$-scaled 
mechanism $\A_x$ satisfies $\sum_i p_i(c)\leq\sum_u t_u\leq\beta(G;x) c(V)$, where 
$\beta(G;x)=\max_{u\in V}\frac{x(N(u))}{x_u}$. Since~\cite{Elkind1} show that 
$\nu(G,c)\geq c(V)/2$, this implies that $\phi_{\A_x}(G)\leq 2\beta(G;x)$.   
Also, if $M$ is a decomposition mechanism constructed from
threshold mechanisms $M_1,\ldots,M_k$, where each $M_q$ satisfies 
$\sum_u t^q_u\leq\phi_q\cdot c(V(G_q))$, then it is easy to see that 
$\phi_M(G)\leq 2\sum_{q=1}^k\phi_q$. 
Thus, we obtain the following results.

\begin{theorem} \label{frugality}
Let $G=(V,E)$ be a graph with $n$ nodes. We can obtain a polytime, truthful, IR mechanism 
$M$ with the following approximation $\rho=\rho_M(G)$ and frugality $\phi=\phi_M(G)$ ratios.
\begin{list}{(\roman{enumi})}{\usecounter{enumi} \topsep=0ex \itemsep=0ex 
    \settowidth{\labelwidth}{(iii)}
    \leftmargin=\labelwidth \addtolength{\leftmargin}{1ex}}
\item $\rho=(\beta(G;x)+1)$, $\phi\leq 2\beta(G;x)$ for \mvcp on $G$;
\item $\rho=O(r^2\log n)$, $\phi=O\bigl(r^2\log n\cdot\Dt(G)\bigr)$ for $r$-dimensional \vcp
on $G$ (using a 2-approximation mechanism with frugality ratio $2\Dt(G)$~\cite{Elkind1}
for single-dimensional \vcp in the construction of Theorem~\ref{rdim}); 
\item $\rho, \phi=O(r\gm\log n)$ for $r$-dimensional \vcp on $G$ when $G$ is everywhere
$\gm$-sparse; hence, we achieve $\rho, \phi=O(r\log n)$ for $r$-dimensional \vcp on any
proper minor-closed family.  
\end{list} 
\end{theorem}

\appendix

\section{Proof of Lemma~\ref{nbredgemap}} \label{append-vc}
Statements (i) and (ii) follow from the statement for general graphs and the
graph-theoretic facts mentioned before Lemma~\ref{nbredgemap}, so we focus on 
proving the statement for an arbitrary graph $G$. Let $\al'=\al'(G)$. 

Consider an arbitrary vertex $v\in V$. For any $u\in N(v)$ define
$x_v^{(uv)}:=\inf \{\sigma\geq 0: t_u(c_v=\beta,c_{-v}=\vec{0})\geq 1 \ \ 
\forall \beta\geq \sigma\}$.

\noindent
{\bf Claim 1}: $x_v^{(uv)}<\infty$. If not, then for any
$p>0$, there exists $q\geq p$ such that
$t_u(c_v=q,c_{-v}=\vec{0})< 1$. So, let $p=\rho$ and $q\geq p$ be such
that $t_u(c_v=q,c_{-v}=\vec{0})< 1$. Consider the cost vector $c$ where $c_u=1$, $c_v=q$,
and $c_z=0$ for $z\neq u,v$, we see that the approximation ratio $\rho$
is contradicted for the instance $(G,c)$ (i.e., graph $G$ with the cost
vector $c$): $V\setminus v$ is an optimal vertex cover of cost $1$ but the threshold
mechanism does not choose $u$ so it chooses $v$ as it is feasible and incurs cost $q>\rho$.

\noindent
{\bf Claim 2}: $x_v^{(uv)}>0$. If $x_v^{(uv)}=0$, then
similar to the above, by considering $c$ where $c_u=1$, $c_v=\e$, $c_z=0$ for $z\neq u,v$,
where $\e$ is very small, we see that $M$ outputs $u$, which means $M$ does not
have the approximation ratio $\rho$.

Now orient the edges of $G$ according to the orientation that determines $\al'(G)$ to
obtain the directed graph $D$. For any arc $(u,v)$ in $D$, consider
linear edge-threshold functions $t_{v}^{(uv)}(c_{u})=
x^{(uv)}_{v}c_{u}$, and $t_{u}^{(uv)}(c_{v})=
(1/x^{(uv)}_{v})c_{v}$. Using these edge-thresholds we
obtain an edge-threshold mechanism $M'$. $M'$ is feasible since for any
arc $(u,v)$ if $u$ is not chosen by $M'$, we should have $c_u>t_{u}^{(uv)}(c_{v})=
(1/x^{(uv)}_{v})c_{v}$ which implies $t_{v}^{(uv)}(c_{u})=
x^{(uv)}_{v}c_{u}>c_v$ hence $v$ is chosen by $M'$.

\begin{figure}[htbp]
\centerline{\includegraphics{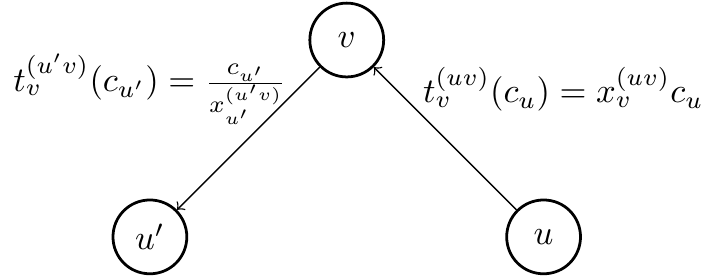}} \label{threshold fig}
\end{figure}

We assert that $M'$ has approximation ratio $O(\rho
\log(\alpha'))$. Note that if $T$ is the outcome of $M'$ and
$T^*$ is the optimal outcome, then we have
\begin{eqnarray*}
c(T) & = & c(T\cap T^*)+c(T\setminus T^*) \leq 
c(T^*)+\sum_{w\in T\setminus T^*}\max_{u\in N(w)}t_w^{(uw)}(c_u) \\
& \leq & c(T^*)+\sum_{w\in T\setminus T^*}\sum_{u\in N(w)}t_w^{(uw)}(c_u) 
 = c(T^*)+\sum_{\substack{w\in T\setminus T^*\\ u\in N(w)}}c_ut_w^{(uw)}(1) \\ 
& = &
c(T^*)+\sum_{u\in T^*} c_u\sum_{w\in N(u)\cap (T\setminus T^*)}t_w^{(uw)}(1) \qquad
(\text{since $N(w)\subseteq T^*$ for $w \notin T^*$})
\end{eqnarray*}

Note that $T\setminus T^*$ is an independent set, so it suffices to show for any $u\in
V(G)$, 
if $S\subseteq N(u)$ forms an independent set then $\sum_{w\in
S}t_w^{(uw)}(1)\leq \rho(\log(\alpha')+2)$.

Let $\delta^{out}(u)=\{v: (u,v)\in D\}$,
$S_1:=S\cap\delta^{out}(u)$, and $S_2:=S\setminus S_1$. So, we have
\begin{align}
\sum_{w\in S}t_w^{(uw)}(1)=\sum_{w\in S_1}t_w^{(uw)}(1)+\sum_{w\in
S_2}t_w^{(uw)}(1)=\sum_{w\in S_1}x_w^{(uw)}+\sum_{w\in S_2}\tfrac{1}{x_u^{(uw)}} \label{three}
\end{align}

Choose an arbitrary $w\in S_1$. By definition of $x_w^{(uw)}$,
for every $\epsilon_w \geq 0$, there is some $0\leq\dt_w\leq\e_w$ such that
$t_u(c_w=x_w^{(uw)}-\epsilon_w+\dt_w,\vec{0})<1$. Hence, $u\notin
M(G,\hat{c})$ where $\hat{c}_w=x_w^{(uw)}-\epsilon_w+\dt_w$, $\hat{c}_u=1$, and
$\hat{c}_z=0$ otherwise. So, since $M(G,\hat{c})$ is a vertex
cover, we should have $w\in M(G,\hat{c})$ which means
$t_w(c_u=1,\vec{0})\geq x_w^{(uw)}-\epsilon_w+\dt_w$. 
Thus, as $S_1$ is an independent set, for the cost vector $c'$ where $c'_u=1$,
$c'_w=x_w^{(uw)}-\epsilon_w+\dt_w$ if $w\in S_1$, and $c'_z=0$
otherwise, we have $S_1\subseteq M(G,c')$ (since $t_w(c'_{N(w)})=t_w(c_u=1,\vec{0})$). 
Letting $\e_w$ tend to 0, we get that $\rho\geq \sum_{w\in S_1}x_w^{(uw)}$, as $V\setminus
N(u)$ is a vertex cover of cost $1$.

Let $S_2=\{v_1,\ldots,v_k\}$ where $x_u^{(uv_1)}\leq
x_u^{(uv_2)}\leq \ldots \leq x_u^{(uv_k)}$.
Consider $c''$ where $c''_u=x_u^{(uv_q)}$, $c''_z=1$ if 
$z\in S_2$, and $c''_z=0$ otherwise. Then, $\{v_1,\ldots,v_q\}\subseteq
M(G,c'')$ hence $\rho \geq q/x_u^{(uv_l)}$ for each
$q=1,\ldots,k$. So, $\sum_{q=1}^k\frac{1}{x_u^{(uv_q)}}\leq
\sum_{q=1}^k\rho/q\leq \rho (\log(|S_2|)+1)\leq \rho
\log(\alpha')+\rho$. Therefore, (\ref{three}) gives
$$
\sum_{w\in S}t_w^{(uw)}(1)\leq \rho + \rho \log(\alpha')+\rho = \rho(\log(\alpha')+2).
\hspace{1.5in} \qedsymbol \hspace{-1.5in}
$$

\section{LP-rounding does not work for \mvcp} \label{lpround-bad}
A common method for designing approximation algorithms for \vcp (and in general) is to 
solve the following LP-relaxation and then round the optimal solution.
\begin{gather}
\min \quad \sum_v c_vx_v \qquad \text{s.t.} \qquad x_u+x_v \geq 1 \quad \forall (u,v)\in E. 
\tag{VC-P} \label{vc-lp} 
\end{gather}
We show that {\em any} LP-rounding algorithm that always includes nodes with
$x_u\geq\frac{1}{2}$ and does not include any node $u$ with $x_u=0$ is not WMON. 

\begin{example} \label{lpround}
Consider the graph $G$ shown below where $u$ and $v$ belong to agent 1. 
For the cost-vector $(c_u,c_a,c_b,c_v,c_d)=(5/4,1,1,1,1)$, the unique optimal solution to
the LP is $(x_u,x_a,x_b,x_v,x_d)=(1/2,1/2,1/2,1/2,1/2)$. Therefore, the algorithm includes
both $u$ and $v$ in the output.
  
\begin{figure}[ht!]
\centerline{\includegraphics{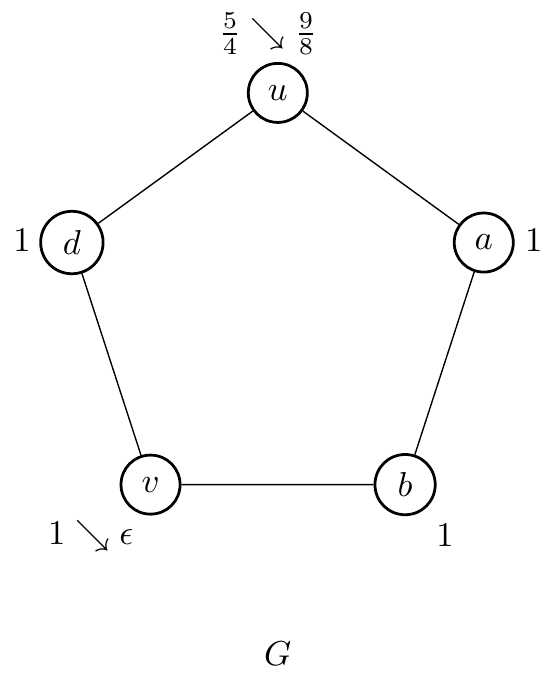}} \label{LP fig}
\end{figure}
 
Consider the cost vector $c'=(c'_1,c_{-1})$ where agent 1 reduces the costs for $u$ and
$v$ to $c'_u=9/8$ and $c'_v=\epsilon < 1/16$ (all other costs are unchanged). Then WMON
dictates that both $u$ and $v$ must still be 
chosen. 
However, the unique optimal solution to the LP with the new costs is
$x_a=x_d=x_v=1,\ x_u=x_b=0$ with cost $2+\epsilon$. (This follows because if $x_u=1$ then
the cost of an LP solution is at least $1+9/8$; if $x_u=1/2$, then the cost of an LP
solution is at least $9/16+1+1/2$; both are greater than $2+\epsilon$ as $\epsilon <
1/16$.) 
So $M$ will not output $u$, which contradicts WMON.
\end{example}

The above example also shows that the following well-known combinatorial 2-approximation
algorithm for \vcp does not satisfy WMON: Given a graph $G=(V,E)$, construct a bipartite
graph $G'$ having two copies of $V$, say $V_1, V_2$, and having edges $(u_1,v_2),
(u_2,v_1)$ for every edge $(u,v)\in E$; solve \vcp on $G'$ and if any of the copies of a
node are chosen in this solution, then pick that node in the solution for $G$. 

In the above example, for the cost-vector $c$, every optimal vertex cover for $G'$
includes exactly one copy of $u$ and one copy of $v$, so both $u$ and $v$ will be chosen
in the solution for $G$. For the cost-vector $c'$, no optimal vertex cover for $G'$
includes any copies of $u$, so $u$ will not be chosen in the solution for $G$. This
contradicts WMON.

\section{Primal-dual methods do not work for \mvcp} \label{other-bad}
The dual of \eqref{vc-lp} is as follows.
\begin{gather}
\max \quad \sum_e y_e \qquad \text{s.t.} \qquad \sum_{e\in\dt(v)}y_e \leq c_v \quad
\forall v\in V.  
\tag{VC-D} \label{vc-dp} 
\end{gather}
Various primal-dual algorithm based on dual ascent are known to yield 2-approximation
algorithms. All of these start with $y=\vec{0}$, raise dual variables while
maintaining dual feasibility, and return the nodes whose costs are completely ``paid''
by the dual variables. 

The two most common variants are where one fixes an ordering of the edges in which to
raise dual variables, and where one raises all (unfrozen) dual variables
simultaneously. We show that neither of these lead to WMON algorithms.
 
\begin{example} \label{pd1}
Consider the graph shown in Fig.~\ref{pd1fig}, where the dual variables are increased in
the order $ux, xy, yv$, and $u$ and $v$ belong to one agent. 

Let $c_u=1,\ c_x=1.5,\ c_y=1.05,\ c_v=0.5$. The primal-dual algorithm will output
$\{u,x,v\}$. Now, if we reduce $c_u$ to $0.5$ and $c_v$ to $0.3$, and keep $c_x$ and $c_y$ 
unchanged, the algorithm outputs $\{u,x,y\}$ which contradicts WMON.

\begin{figure}[ht!]
\centerline{\includegraphics{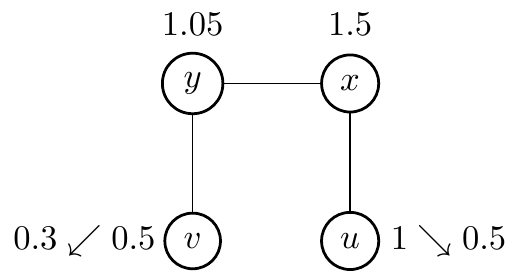}} 
\caption{\label{pd1fig}}
\end{figure}
\end{example}

\begin{example} \label{pd2}
Now consider the simultaneous-dual-ascent primal-dual algorithm.
Consider again the same graph as in Example~\ref{pd1} but with a different assignment of
costs, as shown in Fig.~\ref{pd2fig}. Let $c_u=1,\ c_x=3,\ c_y=4.6,\ c_v=2.5$. The
primal-dual algorithm outputs $\{u,x,v\}$. Now, if we reduce $c_u$ to $0.5$ and $c_v$ to
$2.4$ and keep $c_x$ and $c_y$ unchanged, the algorithm outputs $\{u,y\}$, which
contradicts WMON. 

\begin{figure}[htbp]
\centerline{\includegraphics{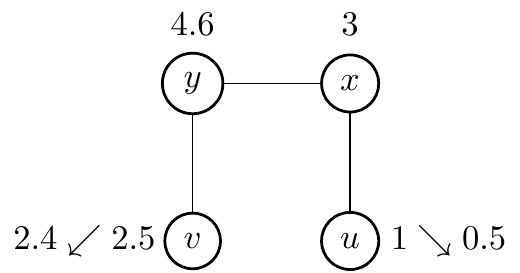}} 
\caption{\label{pd2fig}}
\end{figure}
\end{example}

\end{document}